\documentclass[conference]{IEEEtran}
\IEEEoverridecommandlockouts
\usepackage{balance}
\usepackage{cite}
\usepackage{amsmath,amssymb,amsfonts}
\usepackage{fancyhdr}
\fancypagestyle{cfooter}{ %
\fancyhf{} 
\cfoot{\scriptsize{© 20XX IEEE. Personal use of this material is permitted. Permission from IEEE must be obtained for all other uses, in any current or future media, including reprinting/republishing this material for advertising or promotional purposes, creating new collective works, for resale or redistribution to servers or lists, or reuse of any copyrighted component of this work in other works.}}

}

\usepackage{array}
\usepackage{stfloats}
\usepackage[caption=false,font=footnotesize,position=bottom]{subfig}
\usepackage{url}
\usepackage{graphicx}
\usepackage{textcomp}
\usepackage{xcolor}
\usepackage{tikz}
\usepackage{tikz-3dplot}
\usepackage{ellipsis}
\usetikzlibrary{calc}
\usetikzlibrary{decorations.pathreplacing,decorations.markings,shapes.geometric}
\usetikzlibrary{calc,patterns,angles,quotes,shapes,arrows.meta}
\usetikzlibrary{chains,spy}
\usetikzlibrary{external}
\tikzset{>=Stealth}
\usepackage{cite}
\usepackage{amsmath,amssymb,amsfonts}
\usepackage[capitalise]{cleveref}
\usepackage{amsthm}

\usepackage{algorithm}
\usepackage{algpseudocode}
\usepackage{siunitx}

\usepackage[utf8]{inputenc}

\usepackage{mathtools}
\usepackage{bm}
\usepackage{etoolbox}
\usepackage{scalerel}

\usepackage{pgfplots}
\pgfplotsset{width=5cm,compat=newest}
\usepackage{pgfplotstable}
\usetikzlibrary{fit,backgrounds,fadings}
\pgfdeclarelayer{bg1}
\pgfdeclarelayer{bg2}
\pgfsetlayers{bg1,bg2,main}
\usepackage{filecontents}
\usepackage[font=scriptsize, justification=centering]{caption}
\usepackage{subcaption}
\usepackage{bm}
\usepackage{enumerate}
\usepackage{tabularx}
\usepackage[dvipsnames]{xcolor}
\usepackage{multicol}

\usepackage{fancyhdr}

\newcommand{\vh}{{\bm{h}}}

\newcommand{\vn}{{\bm{n}}}

\newcommand{\vr}{{\bm{r}}}

\newcommand{\vx}{{\bm{x}}}
\newcommand{\vy}{{\bm{y}}}
\newcommand{\vz}{{\bm{z}}}

\newcommand{\mc}{{\bm{C}}}

\newcommand{\mh}{{\bm{H}}}
\newcommand{\mati}{{\bm{I}}}

\newcommand{\mdelta}{{\bm{\Delta}}}
\newcommand{\vgam}{{\bm{\gamma}}}
\newcommand{\veta}{{\bm{\eta}}}

\newcommand{\vmu}{\bm{\mu}}

\newcommand{\vphi}{{\bm{\phi}}}
\newcommand{\mphi}{{\bm{\Phi}}}

\newcommand{\mpsi}{\bm{\Psi}}

\newcommand{\diag}{\mathrm{diag}}

\newcommand{\He}{\mathrm{H}}
\newcommand{\T}{\mathrm{T}}

\renewcommand{\mid}{\, | \,}

\newcommand{\CME}{\mathrm{CME}}
\newcommand{\data}{\mathrm{data}}
\newcommand{\dict}{\mathrm{dict}}
\newcommand{\Frob}{\mathrm{F}}
\newcommand{\GMM}{\mathrm{GMM}}
\newcommand{\NMSE}{\mathrm{NMSE}}
\newcommand{\OMP}{\mathrm{OMP}}
\newcommand{\pilot}{\mathrm{pilot}}
\newcommand{\plOTFS}{\mathrm{pl},\mathrm{OTFS}}
\newcommand{\sparse}{\mathrm{sp}}
\newcommand{\test}{\mathrm{test}}
\newcommand{\train}{\mathrm{train}}

\usepackage{amsthm}
\newtheorem{theorem}{Theorem}

\tikzstyle{block} = [draw, rectangle, 
minimum height=4em, minimum width=4em]
\tikzstyle{input} = [coordinate]
\tikzstyle{output} = [coordinate]
\tikzstyle{pinstyle} = [pin edge={to-,thin,black}]

\usetikzlibrary{positioning}

\tikzset{radiation/.style={{decorate,decoration={expanding waves,angle=90,segment length=5pt}}}}

\usetikzlibrary{spy}
\usepackage{pgfplots}
\usepackage{wrapfig}
\usetikzlibrary{arrows,shapes}
\usetikzlibrary{positioning,shapes.callouts}
\usepgfplotslibrary{groupplots,dateplot}
\usetikzlibrary{patterns,shapes.arrows}
\pgfplotsset{compat=newest}

\def\BibTeX{{\rm B\kern-.05em{\sc i\kern-.025em b}\kern-.08em
		T\kern-.1667em\lower.7ex\hbox{E}\kern-.125emX}}

\usepackage[acronym,shortcuts]{glossaries}
\newacronym{2D}{2D}{two-dimensional}
\newacronym{5G}{5G}{fifth generation}
\newacronym{6G}{6G}{sixth generation}
\newacronym{AWGN}{AWGN}{additive white Gaussian noise}
\newacronym{BS}{BS}{base station}
\newacronym{CGLM}{CGLM}{conditional Gaussian latent models}
\newacronym{CME}{CME}{conditional mean estimator}
\newacronym{CS}{CS}{compressive sensing}
\newacronym{CSGMM}{CSGMM}{compressive sensing Gaussian mixture model}
\newacronym{CSI}{CSI}{channel state information}
\newacronym{CSVAE}{CSVAE}{compressive sensing variational auto-encoder}
\newacronym{DD}{DD}{Doppler-delay}
\newacronym{ELBO}{ELBO}{evidence lower bound}
\newacronym{EM}{EM}{expectation-maximization}
\newacronym{eMBB}{eMBB}{enhanced mobile broadband}
\newacronym{GM}{GM}{generative model}
\newacronym{GMM}{GMM}{Gaussian mixture model}
\newacronym{ICI}{ICI}{inter-carrier interference}
\newacronym{IoT}{IoT}{internet of things}
\newacronym{ISFFT}{ISFFT}{inverse symplectic finite Fourier transform}
\newacronym{ISI}{ISI}{inter-symbol interference}
\newacronym{LMMSE}{LMMSE}{linear minimum mean square error}
\newacronym{LoS}{LoS}{line-of-sight}
\newacronym{MIMO}{MIMO}{multiple-input multiple-output}
\newacronym{ML}{ML}{machine learning}
\newacronym{mmWave}{mmWave}{millimeter-wave}
\newacronym{M-SBL}{M-SBL}{multiple sparse Bayesian learning}
\newacronym{MSE}{MSE}{mean squared error}
\newacronym{NMSE}{NMSE}{normalized mean squared error}
\newacronym{OFDM}{OFDM}{orthogonal frequency-division multiplexing}
\newacronym{OMP}{OMP}{orthogonal matching pursuit}
\newacronym{OTFS}{OTFS}{orthogonal time-frequency space}
\newacronym{PDF}{PDF}{probability density function}
\newacronym{SBGM}{SBGM}{sparse Bayesian generative model}
\newacronym{SBL}{SBL}{sparse Bayesian learning}
\newacronym{SFFT}{SFFT}{symplectic finite Fourier transform}
\newacronym{SNR}{SNR}{signal-to-noise ratio}
\newacronym{TF}{TF}{time-frequency}
\newacronym{UE}{UE}{user entity}
\newacronym{ULA}{ULA}{uniform linear array}
\newacronym{VAE}{VAE}{variational auto-encoder}
\newacronym{WSSUS}{WSSUS}{wide-sense-stationary-uncorrelated-scattering}

\tikzset{PlotDFT/.style={mark=star,mark size=2.2pt, line width=1pt, color=gray, dashed, mark options=solid}}
\tikzset{PlotCVAE-single/.style={mark=triangle,mark size=1.5pt, line width=1pt, color=green!40!black, dashed, mark options={solid, rotate=180}}}
\tikzset{PlotGeo/.style={mark=triangle,mark size=1.5pt, line width=1pt, color=black, dashed, mark options={solid, rotate=0}}}
\tikzset{PlotCVAE-no-cond/.style={mark=diamond,mark size=1.8pt, line width=1pt, dashed, color=black!50!green, mark options=solid}}
\tikzset{PlotCVAE-genie/.style={mark=x, mark size=2.2pt, domain=1:10000, line width=1pt, color=red, dashed, mark options=solid}}
\tikzset{PlotCVAE-cond/.style={mark=square,mark size=1.5pt, line width=1pt, color=blue, mark options=solid}}
\tikzset{PlotCGMM-S/.style={mark=o, mark size=1.5pt, line width=1pt, color=orange , dashed, mark options=solid}}

\usepackage{circuitikz}

\begin{document}

\title{
OTFS Channel Estimation Utilizing Sparse Bayesian Generative Modelling
\thanks{This work is funded by the Bavarian Ministry of Economic Affairs, Regional Development, and Energy within the project 6G Future Lab Bavaria.}
}
\newcommand{\jbig}{$\mathcal{J}$}
\author{\IEEEauthorblockN{Louis Anseaume, Benedikt Böck, Franz Weißer, Wolfgang Utschick\\}
\IEEEauthorblockA{\textit{TUM School of Computation, Information and Technology, Technical University of Munich, Germany} \\
\{louis.anseaume,benedikt.boeck,franz.weisser,utschick\}@tum.de}
}

\maketitle
\thispagestyle{cfooter}


\begin{abstract}
    One of the key challenges of future wireless communication systems is ensuring reliability in high-speed mobile scenarios, where accurate recovery of \ac{CSI} is essential. Many recent studies have concluded that \ac{OTFS} modulation is a promising technology for addressing this challenge. Additionally, \ac{ML}-based methods have the potential to improve channel estimation performance by leveraging ambient information more effectively than classical estimation techniques. This paper particularly addresses channel estimation for \ac{OTFS} by employing a \ac{CS}-based \ac{SBGM}, namely the recently introduced \ac{CSGMM}. We show that our proposed approach yields significant improvement in \ac{NMSE} over the next-best-performing baseline. We additionally provide insights into the theoretical potential of the model to optimally approximate complex channel distributions with arbitrary precision within the \ac{DD} domain. To summarize, this work establishes the \ac{OTFS}-\ac{CSGMM} framework as a promising solution for high-mobility wireless channel estimation.
\end{abstract}


\glsresetall 
\section{Introduction}

Current state-of-the-art wireless communication systems, such as \ac{5G} systems, are based on \ac{OFDM} due to its high spectral efficiency, low computational complexity, and robustness to frequency-selective fading \cite{Stuber:2004}. However, next-generation systems, such as \ac{6G} systems, aim to expand the network to support high-mobility scenarios, such as vehicle-to-vehicle or high-speed train communication, with speeds reaching between 300 and 500 km/h in typical applications \cite{Yuan:2023}. In these contexts, research has shown that \ac{OFDM} no longer remains a reliable technology due to high \ac{ICI}, leading to poor signal reconstruction \cite{Wang:2006}. Many recent works have thus declared that \ac{OTFS} is a promising candidate technology for these scenarios involving large Doppler effects \cite{Hadani:2017}. It operates in the \ac{DD} domain, a \ac{2D} space spanning Doppler shifts and delays, instead of passing the information symbols over the usual \ac{TF} domain as in \ac{OFDM}. Accurate \ac{CSI} is crucial for \ac{OTFS} to achieve correct, reliable, and fast communication, and therefore requires highly accurate channel estimation methods.

One important aspect that can be utilized to obtain \ac{CSI} in \ac{OTFS} modulation is the channel's inherent sparsity in typical outdoor high-mobility scenarios, due to the presence of a small number of reflectors \cite{Wang:2016}. However, classical estimation methods for sparse channel representations may reach their limits in these complex environments. Therefore, \ac{ML}-based methods are promising candidates for improving channel estimation in these scenarios, as they can better leverage contextual information inherent to the wireless communication system. In particular, \cite{Boeck:2024} introduces the \ac{CS}-based \ac{SBGM}, which can be used to learn the channel distribution within a base station cell in typical high-speed mobile communication setups.

The goal of this work is to address the problem of \ac{OTFS} channel estimation by leveraging the sparsity of typical outdoors wireless communication environments, as well as the inherent adaptability of \ac{ML} techniques to capture the complexity of these scenarios. To this end, we reformulate the results from \cite{Gaudio:2022} in order to match the description of \ac{CS}. Afterwards, we utilize the \ac{CSGMM} from \cite{Boeck:2024}, a representative of \acp{SBGM}, and apply the adapted algorithm to the problem of \ac{OTFS} channel estimation. Through simulations, we show that our proposed approach outperforms state-of-the-art baseline methods. We additionally provide a theoretical study of the functionalities of \acp{SBGM}, showing that, under certain assumptions on the channels, the studied models can optimally learn the characteristics of the underlying channel distributions, or arbitrarily fine approximations thereof. This will allow the estimation methods to yield optimal channel reconstructions.
\section{OTFS System Model}

We consider time-frequency-dependent wireless channels, consisting of $P$ multipath components. The channel impulse response, within the \ac{TF} domain, is then expressed as
\begin{equation}
H(t,f) = \sum_{p=0}^{P-1} h_p e^{j2\pi \nu_p t} e^{-j2\pi \tau_p f},
\label{eq:TF-channel-model}
\end{equation}
where $\{h_p\}_{p=0}^{P-1}$ are the channel gains, $\{\nu_p\}_{p=0}^{P-1}$ the Doppler shifts and $\{\tau_p\}_{p=0}^{P-1}$ the delays of the paths. Furthermore, we denote with $T$ the symbol duration within the \ac{TF} domain and $\Delta f$ the subcarrier frequency spacing. We assume that the channel parameters follow~\cite{Gaudio:2022}
\begin{equation}
\tau_{\max} < T, \hspace{0.5cm} \nu_{\max} < \Delta f, \hspace{0.5cm} T = 1/\Delta f,
\end{equation}
where $\tau_{\max}$ and $\nu_{\max}$ represent the (absolute) maximal values that the parameters $\tau_p$ and $\nu_p$ can take for $p \in \left\{ 0,...,P-1 \right\}$, respectively.

A set of complex data symbols $\left\{x_{k^\prime,l^\prime}\right\}$ for $k^\prime \in \left\{ 0,...,N-1 \right\}$ and $l^\prime \in \left\{ 0,...,M-1 \right\}$, with $N,M \in \mathbb{N}$, is placed in a \ac{2D} \ac{DD} grid 
\begin{equation}
\mathcal{G}_{\data} = \left\{ \left(\frac{k^\prime}{NT},\frac{l^\prime}{M\Delta f} \right)\right\}_{k^\prime,l^\prime = 0}^{N-1,M-1},
\label{eq:DD-modulation-grid}
\end{equation}
which we call the modulation grid. The relation between the \ac{DD} domain and the \ac{TF} domain is then given by the \ac{ISFFT}, which applied to $\left\{x_{k^\prime,l^\prime}\right\}$ yields a block of data symbols $\left\{X[n^\prime,m^\prime]\right\}$ in the \ac{TF} domain. By applying a discrete Heisenberg transform on the \ac{TF} symbols, we obtain a transmission signal in the time domain, which is transmitted over the time-frequency-selective channel defined in \eqref{eq:TF-channel-model}. At the receiver, we apply a Wigner transform on the received signal to obtain the output signal $y(t,f)$ in the \ac{TF} domain. After sampling the output signal to obtain a block of symbols $\left\{Y[n,m]\right\}$ in the \ac{TF} domain, we apply the \ac{SFFT}, yielding output symbols $\left\{y_{k,l}\right\}$ in the \ac{DD} domain.

Following the described modulation scheme, the vectorized input-output relationship for \ac{OTFS} is then \cite{Gaudio:2022}
\begin{equation}
\vy = \left( \sum_{p=0}^{P-1}h_p \mpsi_{p} \right) \vx + \vn = \mh \vx + \vn,
\label{eq:OTFS-input-output}
\end{equation}
where $\mh$ is the channel matrix, the entries of $\mpsi_p$ are the \ac{ISI} coefficients $\Psi_{k,k^\prime}^p [l,l^\prime]$ of the $p$-th path, and $\vn\sim\mathcal{N}_\mathbb{C}(\bm{0},\sigma^2 \mati_{MN})$ is \ac{AWGN}. By considering rectangular shaping pulses for both transmission and reception, for any path $p \in \left\{ 0,...,P-1 \right\}$, a simplified expression of $\Psi_{k,k^\prime}^p[l,l^\prime]$ is given by \eqref{eq:ISI-mat-element}, shown at the bottom of this page. 
The details of the derivation can be found in \cite{Gaudio:2020}. 
\begin{figure*}[b]
\hrule
\begin{equation}
\begin{split}
    \Psi_{k,k^\prime}^p[l,l^\prime] \approx \frac{e^{j2\pi \nu_p \tau_p}}{NM}\frac{ 1-e^{j2\pi (k^\prime - k + \nu_pNT)} }{ 1-e^{j2\pi \frac{(k^\prime - k + \nu_pNT)}{N}} }\frac{ 1-e^{j2\pi (l^\prime - l + \tau_pM\Delta f)} }{ 1-e^{j2\pi \frac{(l^\prime - l + \tau_pM\Delta f)}{M}} }e^{-j2\pi\nu_p\frac{l^\prime}{M\Delta f}} \\
    \times \begin{cases}
        1 & \text{if } l^\prime \in \left\{ 0,...,M-1-\left\lceil \frac{M\tau_p}{T} \right\rceil \right\} \\
        e^{-j2\pi\left( \frac{k^\prime}{N} + \nu_pT \right)} & \text{if } l^\prime \in \left\{ M-\left\lceil \frac{M\tau_p}{T} \right\rceil ,...,M-1 \right\}.
    \end{cases}
\end{split}
\label{eq:ISI-mat-element}
\end{equation}
\end{figure*}
The channel has a block-wise effect on the \ac{OTFS} modulation scheme instead of inducing symbol-per-symbol transformations \cite{Gaudio:2022}.
The transmitted symbols experience a shift in the \ac{DD} grid, corresponding to the channel parameters, with some additional energy spread induced by their fractional parts.
This leads us to consider the pilot scheme from \cite{Gaudio:2022}, which we summarize in the following.

Here, a block of $N \times M$ symbols contains a central, high-energy pilot symbol, called a peak pilot symbol, surrounded by a zero-energy guard interval of pilot symbols. The remaining space in the block can be reserved for low-energy data symbols. This ensures that the peak pilot symbol is isolated from the data symbols, both spatially within the \ac{DD} domain and in power.
Therefore, it can serve as an indicator of shifts in the \ac{DD} grid as interference with other symbols is negligible. As such, we can set the input vector of such a pilot scheme as
\begin{equation}
\vx_{\plOTFS}=[0,...,0,x_{\pilot},0,...,0]^{\T} \in \mathbb{C}^{NM}
\end{equation}
since data symbols have negligible power levels.
\section{OTFS Channel Estimation}

We consider a class of models that incorporates both the sparsity-inducing property of the \ac{SBL} framework \cite{Wipf:2004} as well as the exploitation of a latent space through a trainable \ac{GM}. In this work, we consider a specific class of \ac{GM}s called \ac{CGLM}, where the distributions conditioned on the latent variable are Gaussians. In that case, the combined models are called \ac{SBGM}s \cite{Boeck:2024}.

In particular, this paper focuses on the application of the \ac{GMM}, which approximates an arbitrary \ac{PDF} $p(\vx)$ by
\begin{equation}
p^{(\GMM)}(\vx) = \sum_{k=1}^K p(k)p(\vx\mid k) = \sum_{k=1}^K \rho_k\mathcal{N}(\vx; \vmu_k, \mc_k)
\label{eq:GMM}
\end{equation}
with $p(k)=\rho_k \geq 0$ for all $k \in \left\{ 1,...,K \right\}$ and $\sum_{k=1}^K\rho_k = 1$. \ac{GMM}s are typically thought to be good approximations of complex probabilistic distributions and can be employed for a variety of wireless communication tasks \cite{Koller:2022}. Incorporating them into the \ac{SBGM} framework yields the \ac{CSGMM} \cite{Boeck:2024}.

To estimate the channel parameters $\xi = \{ h_p,\nu_p,\tau_p \}_{p=0}^{P-1}$ utilizing the framework of the \ac{CSGMM}, we reformulate the \ac{OTFS} input-output relationship. To this end, we introduce the vectors $\vphi_p$ as 
\begin{equation}
\vphi_p = \mpsi_p \vx.
\label{eq:phi_vector}
\end{equation}
With the matrix $\mphi = \left[ \vphi_0,...,\vphi_{P-1} \right] \in \mathbb{C}^{MN \times P}$ and the vector $\vh = \left[ h_0,...,h_{P-1} \right]^{\T} \in \mathbb{C}^P$, the input-output relationship in~\eqref{eq:OTFS-input-output} becomes
\begin{equation}
\vy = \mphi\vh + \vn.
\label{eq:OTFS-input-output-mat}
\end{equation}
Let us further define a \ac{DD} grid of cardinality $G \geq MN$, which we call the sparse grid, as
\begin{equation}
\mathcal{G} = \left\{ \left( \tilde{\nu}_i, \tilde{\tau}_i \right) \right\}_{i=1}^{G}.
\label{eq:DD-sparse-grid}
\end{equation}
Under the assumption that the channel coefficients lie on the sparse grid, the input-output relationship in~\eqref{eq:OTFS-input-output-mat} can be reformulated as
\begin{equation}
\vy = \mphi_{\dict}\vh_{\sparse} + \vn,
\label{eq:OTFS-input-output-sparse}
\end{equation} 
with the sparse vector $\vh_{\sparse} \in \mathbb{C}^G$, which has only $P$ non-zero entries, and the sparse dictionary $\mphi_{\dict}$, which is given as
\begin{equation}
\mphi_{\dict} = \left[ \tilde{\vphi}_0,...,\tilde{\vphi}_{G-1} \right] \in \mathbb{C}^{MN \times G},
\label{eq:DD-sparse-dict}
\end{equation}
with its columns as $\tilde{\vphi}_i = \mpsi\left( \tilde{\nu}_i, \tilde{\tau}_i \right) \vx$, cf. \eqref{eq:phi_vector}.
In the general case, we still use \eqref{eq:OTFS-input-output-sparse} under the assumption that for any parameters $\xi = \{ h_p,\nu_p,\tau_p \}_{p=0}^{P-1}$ the approximation $\mphi\vh \approx \mphi_{\dict}\vh_{\sparse}$ holds. 
Equation \eqref{eq:OTFS-input-output-sparse} now resembles the typical formulation of a CS ill-posed linear inverse problem \cite{Boeck:2024}. This allows us to express the \ac{OTFS}-\ac{CSGMM} as follows,
\begin{equation}
\begin{cases}
\vy\mid\vh_{\sparse} \sim p(\vy\mid\vh_{\sparse}) = \mathcal{N}(\vy; \mphi_{\dict}\vh_{\sparse}, \sigma^2 \mati) \\
\vh_{\sparse}\mid k \sim p(\vh_{\sparse}\mid k) = \mathcal{N}(\vh_{\sparse}; \bm{0}, \text{diag}(\vgam_k)) \\
\forall k \in \left\{ 1,...,K \right\}, p_\delta(k)=\rho_k \geq 0.
\end{cases}
\label{eq:OTFS-CSGMM}
\end{equation}
We use an \ac{EM} algorithm to train the \ac{OTFS}-\ac{CSGMM}. In the so-called E-step, we estimate the statistical characteristics of the posteriors of the sparse and latent variables $\{\vh_{\sparse}\mid(\vy,k); k\mid\vy \}$, which allows the closed-form computation of $p_{(t)}(\vh_{\sparse}, k\mid\vy)$ at the $t$-th iteration. In the so-called M-step, we update the parameters of the \ac{CSGMM} $\{ \rho_{k,(t+1)}, \vgam_{k,(t+1)}\}_{k=1}^K$ by maximizing the \ac{ELBO} with respect to these parameters.

After the training of the model, we estimate the sparse channel representations $\vh_{\sparse}$ from observed noisy samples $\vy$.
The \ac{MSE}-optimal \ac{CME} can be decomposed as
\begin{equation}
\begin{split}
    \vh_{\sparse,\CME} & = \mathbb{E} \left[ \vh_{\sparse}\mid\vy \right] = \mathbb{E} \left[ \mathbb{E} \left[ \vh_{\sparse}\mid \vy, k \right]\mid \vy \right] \\
    & \approx \sum_{k=1}^K \hat{p}(k\mid\vy) \hat{\vmu}^{\vh_{\sparse}\mid\vy, k}
\end{split}
\label{eq:CME}
\end{equation}
where $\hat{p}(k\mid\vy)$ and $\hat{\vmu}^{\vh_{\sparse}\mid\vy, k}$ are the estimates of the respective quantities. This sparse estimate is then cast into its matrix representation
\begin{equation}
    \hat{\mh} = \sum_{i=1}^{G} h_{i, {\sparse}}\mpsi\left( \tilde{\nu}_i,\tilde{\tau}_i \right)
\end{equation}
where $h_{i, {\sparse}}$ are the entries of $\vh_{\sparse,\CME}$.

\section{Theoretical Analysis}

Under the assumption of a zero mean and \ac{WSSUS} channel model, the statistical characteristics are preserved under conditioning on side information \cite{Boeck-statchar:2024}. However, for \ac{OTFS}, the Toeplitz structure constraint arising from the \ac{WSSUS} assumption no longer holds. Nevertheless, we show that it is possible to design a procedure in which, under suitable conditions, the \ac{CSGMM} can optimally learn arbitrarily fine approximations of the channel distribution, which allows accurate recovery of the distribution of random \ac{OTFS} channels.

Let us consider the decomposition of the channel path coefficients as $h_p = \sqrt{a_p} e^{-j\beta_p}$ for $p \in \left\{ 0,...,P-1 \right\}$, as well as the resulting channel model
\begin{equation}
    \vr = \mphi\vh = \sum_{p=0}^{P-1} \sqrt{a_p} e^{-j\beta_p}\vphi_{p}.
    \label{eq:vectorized-OTFS-model}
\end{equation}

\subsection{Discrete Model Distribution}
If the channel parameters lie on the sparse grid $\mathcal{G}$, the following theorem shows that the statistical characterization of the \ac{OTFS} model is preserved under conditioning on side information.

\begin{theorem}
    Let $\mathcal{G}$ be a \ac{DD} grid such that for all $p \in \left\{ 0,...,P-1 \right\}$, we have $\left( \nu_p, \tau_p \right) \in \mathcal{G}$, and let $\mphi_{\dict}$ be the associated sparse dictionary (cf. \eqref{eq:DD-sparse-dict}). Let $\vz$ be any side information about $\vr$. Assume that
    \vspace{-0.4cm}
    \begin{multicols}{2}
        \begin{itemize}
            \item $\forall p \neq p^\prime, \beta_p \perp \beta_{p^\prime}$
            \item $\forall p, \beta_p \perp a_p, \nu_p, \tau_p$
            \item $\forall p, \beta_p \sim \mathcal{U}([-\pi, \pi])$
            \item $\forall p, \beta_p \perp \vz$
        \end{itemize}
    \end{multicols}
    \vspace{-0.4cm}
    \noindent
    Then $\vr$ has the following moments,
    \begin{equation}
        \mathbb{E}\left[ \vr \right] = \bm{0}
        \label{eq:th-1-mean}
    \end{equation}
    \begin{equation}
        \mathbb{E}\left[ \vr\vr^{\He} \right] = \mphi_{\dict}\mdelta\mphi_{\dict}^{\He}
        \label{eq:th-1-cov}
    \end{equation}
    \begin{equation}
        \mathbb{E}\left[ \vr\mid\vz \right] = \bm{0}
        \label{eq:th-1-mean-cond}
    \end{equation}
    \begin{equation}
        \mathbb{E}\left[ \vr\vr^{\He}\mid\vz \right] = \mphi_{\dict}\mdelta_{\vz}\mphi_{\dict}^{\He}
        \label{eq:th-1-cov-cond}
    \end{equation}
    where $\mdelta$ and $\mdelta_{\vz}$ are diagonal matrices.
    \label{th:1}
\end{theorem}

\begin{proof}
    See Appendix.
\end{proof}

According to Theorem \ref{th:1}, if certain assumptions on the side information $\vz$ and the path phases $\beta = \left\{\beta_p\right\}_{p=0}^{P-1}$ hold,
the zero-mean property and dictionary-based structure of the model are preserved.
Thus, the model can, in principle, optimally learn this distribution since the ground truth model distribution lies within the search space of the \ac{CSGMM} \cite{Boeck-statchar:2024}.

\subsection{Continuous Model Distribution}

In general, the parameters of ground truth channels are continuously distributed over a bounded range of possible values. 
By constructing a sequence of discrete models, the following theorem establishes convergence to the continuous ground truth model.

\begin{theorem}
    Let $\left(\vr_{(n)}\right)_{n \in \mathbb{N}}$ be a sequence of discrete models defined on a sequence of increasingly dense, evenly spaced \ac{DD} grids $\mathcal{G}_{(n)}$. We write $\mphi_{\dict ,(n)}$ as the sparse dictionary associated to \ac{DD} grid $\mathcal{G}_{(n)}$. Assume boundedness of the \ac{DD} parameters. Let $\vz$ be some side information on $\vr$. Furthermore, assume
    \vspace{-0.4cm}
    \begin{multicols}{2}
        \begin{itemize}
            \item $\forall p \neq p^\prime, \beta_p \perp \beta_{p^\prime}$
            \item $\forall p, \beta_p \perp a_p, \nu_p, \tau_p$
            \item $\forall p, \beta_p \sim \mathcal{U}([-\pi, \pi])$
            \item $\forall p, \beta_p \perp \vz$
        \end{itemize}
    \end{multicols}
    \vspace{-0.4cm}
    \noindent
    Then
    \begin{equation}
        \forall n \in \mathbb{N}, \mathbb{E}\left[ \vr_{(n)} \right] = \mathbb{E}\left[ \vr_{(n)}\mid\vz \right] = \bm{0}
        \label{eq:th-2-mean-seq}
    \end{equation}
    \begin{equation}
        \mathbb{E}\left[ \vr \right] = \mathbb{E}\left[ \vr\mid\vz \right] = \bm{0}
        \label{eq:th-2-mean-cont}
    \end{equation}
    \begin{equation}
        \mphi_{\dict ,(n)}\mdelta_{(n)}\mphi_{\dict ,(n)}^{\He} = \mathbb{E}\left[ \vr_{(n)}\vr_{(n)}^{\He} \right] 
        \hspace{-2pt}\underset{n \rightarrow +\infty}{\xrightarrow{\hspace{16pt}}}\hspace{-2pt}
        \mathbb{E}\left[ \vr\vr^{\He} \right]
        \label{eq:th-2-cov}
    \end{equation}
    \begin{equation}
        \mphi_{\dict ,(n)}\mdelta_{(n),\vz}\mphi_{\dict ,(n)}^{\He} = \mathbb{E}\left[ \vr_{(n)}\vr_{(n)}^{\He}\mid\vz \right] 
        \hspace{-2pt}\underset{n \rightarrow +\infty}{\xrightarrow{\hspace{16pt}}}\hspace{-2pt}
        \mathbb{E}\left[ \vr\vr^{\He}\mid\vz \right]
        \label{eq:th-2-cov-cond}
    \end{equation}
    where $\mdelta_{(n)}$ is a diagonal matrix dependent on the \acs{DD} grid $\mathcal{G}_{(n)}$, and $\mdelta_{(n),\vz}$ is a diagonal matrix that additionally depends on the side information $\vz$.
    \label{th:2}
\end{theorem}

\begin{proof}
    \eqref{eq:th-2-mean-seq} and \eqref{eq:th-2-mean-cont} stem from similar computations to the results in \eqref{eq:th-1-mean} and \eqref{eq:th-1-mean-cond}. \eqref{eq:th-2-cov} and \eqref{eq:th-2-cov-cond} are the result of the dominated convergence theorem.
\end{proof}

Theorem \ref{th:2} shows that the statistical characteristics of discretely parameterized models converge to those of the ground truth continuous channel. 
Since such discrete models lie within the learnable space of the \ac{CSGMM}, this result demonstrates that the \ac{CSGMM} can, in principle, capture arbitrarily accurate representations of realistic \ac{OTFS} channels. Therefore, it provides theoretical motivation for implementing the \ac{CSGMM} with high-resolution \ac{DD} grids, ensuring that the moments of the learned approximations approach those of the ground truth, continuously parameterized \ac{OTFS} models.
\section{Numerical Results}


\begin{figure}[t]
\includegraphics{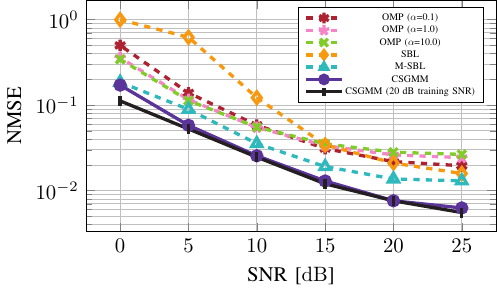}
\caption{\acs{NMSE} performance of \acs{CSGMM}, \acs{OMP}, \acs{SBL}, and \acs{M-SBL} applied to \acs{OTFS} channel estimation for different \acs{SNR} values. For trained methods, at each tested \acs{SNR} level, we separately train a new model at that \acs{SNR}.}
\label{fig:OTFS-CSGMM-vs-baselines-over-SNR}
\vspace{-0.4cm}
\end{figure}

To evaluate our proposed channel estimation framework, we consider the ``Dynamic Doppler (DD1)" scenario from DeepMIMOv3~\cite{Alkhateeb:2019} with the \ac{BS} labeled ``BS1''.
Furthermore, we use single-element antennas and generate 
the $P=10$ most significant paths for \ac{BS}-to-user downlink channels across the 2000 available scenes.
We randomly shuffle and partition the dataset into a training and a testing set with sizes $N_{\train} = 12805$ and $N_{\test} = 3202$, respectively.

The implementation of the \ac{OTFS} modulation scheme is done at a subcarrier frequency spacing of $\Delta f = \SI{30}{\kilo Hz}$ and a symbol duration $T=\frac{1}{\Delta f}= \SI{33.3}{\micro\second}$. For the \ac{CSGMM}, we consider $K=10$ components. The \ac{OTFS} modulation \ac{DD} grid $\mathcal{G}_{\data}$ is set to be an $8 \times 8$ grid, while the \ac{OTFS} sparse \ac{DD} grid $\mathcal{G}$ spans $32 \times 32$ elements, unless specified otherwise, such that we have $|\mathcal{G}_{\data}|=64$ and $|\mathcal{G}|=1024$. Except where otherwise provided, we set the \ac{SNR} at \SI{20}{dB}.

The \ac{NMSE} as our distortion metric is given in its empirical form by
\begin{equation}
\NMSE = \frac{1}{N_{\test}} \sum_{n=1}^{N_{\test}} \frac{\left\Vert \hat{\mh} - \mh \right\Vert_{\Frob}^2}{\left\Vert \mh \right\Vert_{\Frob}^2}.
\label{eq:NMSE}
\end{equation}
For comparative baselines, we consider \ac{OMP} \cite{Pachigolla:2025}, \ac{SBL} \cite{Wipf:2004} and its variant \ac{M-SBL} \cite{Wipf:2007}. In our implementation of \ac{OMP}, the tolerance level of the residue norm is defined as $\epsilon_{\OMP} = \alpha|\mathcal{G}|$, where $\alpha$ denotes a scaling factor and $|\mathcal{G}|$ the grid size.
From \cite{Wipf:2007}, we see that we can equate \ac{M-SBL} to a single-component \ac{CSGMM} model, cf. \eqref{eq:OTFS-CSGMM}.

In the first experiment, we train and test models separately for each SNR value in a specified range, retraining each learning-based model independently at every assessed SNR level. We additionally evaluate the generalization capability of our approach by testing a \ac{CSGMM} trained at a single, fixed SNR. 
The corresponding performance results are presented in Fig.~\ref{fig:OTFS-CSGMM-vs-baselines-over-SNR}. 
The \ac{CSGMM} consistently outperforms all baseline models across the entire \ac{SNR} range, showing a significant improvement over the next-best-performing method, namely \ac{M-SBL}. At high \ac{SNR} levels, all implemented methods exhibit a flattening in performance, suggesting saturation. These observations indicate that the proposed \ac{CSGMM} more accurately captures the probabilistic structure of \ac{OTFS} channels compared to existing dictionary-based approaches. The reduced improvement in distortion at high \ac{SNR} values further highlights the limitations introduced by fractional Doppler effects when using discretized \ac{DD} grid approximations. While its accuracy remains ultimately constrained by the resolution of the discrete grid representation of the channel distribution in the \ac{DD} domain, leading to a residual generalization bias, the \ac{CSGMM} demonstrates strong generalization performance, making it a promising candidate for \ac{OTFS} channel estimation.

\begin{figure}[t]
\includegraphics{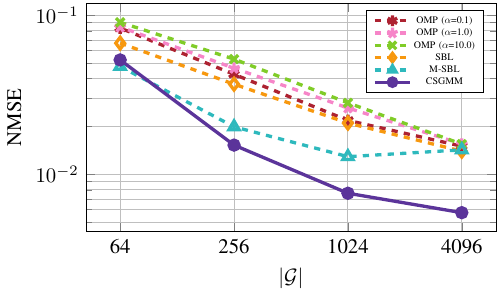}
\vspace{-0.13cm}
\caption{\acs{NMSE} performance of \acs{CSGMM}, \acs{OMP}, \acs{SBL}, and \acs{M-SBL} applied to \acs{OTFS} channel estimation for different sparse grid sizes $|\mathcal{G}|$.}
\label{fig:OTFS-CSGMM-vs-baselines-over-sdim}
\vspace{-0.4cm}
\end{figure}

To elaborate more on the impact of the discretized grid, in the second experiment, the \ac{SNR} is fixed and the performance of the considered methods is evaluated as a function of the sparse grid cardinality, $|\mathcal{G}|$. The corresponding results are shown in Fig.~\ref{fig:OTFS-CSGMM-vs-baselines-over-sdim}. For large grids, the \ac{CSGMM} consistently outperforms all baselines, achieving significant improvement for a $64 \times 64$ grid relative to the best competing method, \ac{SBL}. 
The direct estimation baselines also show improved performance as the grid size increases, but keep their relatively inferior performance. 
Furthermore, while the \ac{CSGMM} benefits monotonically from larger $|\mathcal{G}|$, the \ac{M-SBL} model attains its peak performance at $|\mathcal{G}| = 1024$ and degrades thereafter, suggesting over-regularization. 
By comparison, the \ac{CSGMM} possesses sufficient model complexity to effectively exploit increased sparse grid dimensionality, thereby consistently benefiting from larger grids. 
Overall, the results highlight the \ac{CSGMM}'s superior ability to capture realistic \ac{OTFS} channel distributions over sparse \ac{DD} grids.

\section{Conclusion}


This paper proposes a novel way to address the problem of \ac{OTFS} channel estimation through a statistically expressive \ac{ML}-based \ac{SBGM}. The used \ac{CSGMM} accurately captures the complexities of channel distributions in high-mobility wireless communication scenarios. We have demonstrated significant performance improvements over state-of-the-art baseline methods in terms of estimation accuracy.
Additionally, we have provided theoretical insights into the potential of our approach to learn arbitrarily accurate approximations of \ac{OTFS} channel distributions. The application of the \ac{CSGMM}, as well as the related model, the \ac{CSVAE}, to the problem of \ac{OTFS} channel estimation offers an exciting opportunity for future research.
\section{Appendix}


\begin{proof}[Proof of Theorem \ref{th:1}]
    Let $\vr$ be a random vector defined according to \eqref{eq:vectorized-OTFS-model}, where its parameters follow all assumptions. Let $G =|\mathcal{G}|$ be the cardinality of $\mathcal{G}$. Since the parameters lie in $\mathcal{G}$, we can write $\vr = \mpsi_{\dict}\vh_{\sparse}$, where we can decompose the entries of $\vh_{\sparse}$ as $\vh_{\sparse, i} = \sqrt{\eta_i}e^{-j\tilde{\beta}_i}$ for $i \in \left\{ 1,...,G \right\}$. Without loss of generalization, we can assume that all elements of $\tilde{\beta} = \left\{ \tilde{\beta}_i \right\}_{i=1}^G$ have the same properties as the elements from $\beta = \left\{ \beta_p \right\}_{p=0}^{P-1}$. It is then easily shown that $\forall i \in \left\{ 1,...,G \right\}, \tilde{\beta}_i \perp \eta_i$. As such, we can perform a change of variables on the moments of $\vr$ applied in \eqref{eq:th-1-mean-proof} and \eqref{eq:th-1-cov-proof}, shown at the bottom of this page, which yields \eqref{eq:th-1-mean} and \eqref{eq:th-1-cov} respectively. By additionally assuming $\forall p, \beta_p \perp \vz$, without loss of generalization, we can assume that $\forall i, \tilde{\beta}_i \perp \vz$, and through a similar reasoning, we obtain \eqref{eq:th-1-mean-cond} and \eqref{eq:th-1-cov-cond}. 
\end{proof}

\begin{figure*}[b]
\hrule
\vspace{0.2cm}
\begin{equation}
    \mathbb{E}\left[ \vr \right] = \mathbb{E}_{\veta, \tilde{\beta}} \left[ \vr \right] = \sum_{i=1}^{G}\mathbb{E}_{\veta}\left[\sqrt{\eta_i}\right]\underbrace{\mathbb{E}_{\tilde{\beta}}\left[e^{-j\tilde{\beta}_i}\right]}_{=0 \text{ since } \tilde{\beta}_i \sim \mathcal{U}([-\pi, \pi])} \vphi\left(\tilde{\nu}^{(i)}, \tilde{\tau}^{(i)} \right) = \bm{0}
\label{eq:th-1-mean-proof}
\end{equation}
\vspace{-0.2cm}
\begin{equation}
\begin{split}
    \mathbb{E}\left[ \vr\vr^{\He} \right] & = \mathbb{E}_{\veta, \tilde{\beta}} \left[ \vr\vr^{\He} \right] = \sum_{i=1}^{G}\sum_{i'=1}^{G}\mathbb{E}_{\veta}\left[\sqrt{\eta_i}\sqrt{\eta_{i'}}\right]\underbrace{\mathbb{E}_{\tilde{\beta}}\left[e^{-j\tilde{\beta}_i}e^{j\tilde{\beta}_{i'}}\right]}_{=\delta_{i,i'} \text{ since } \tilde{\beta}_i \perp \tilde{\beta}_{i'}} \vphi\left(\tilde{\nu}^{(i)}, \tilde{\tau}^{(i)}\right)\vphi^{\He}\left(\tilde{\nu}^{(i^\prime)}, \tilde{\tau}^{(i^\prime)}\right)\\
\vspace{-0.2cm}
    & = \sum_{i=1}^{G}\mathbb{E}_{\veta}\left[\eta_i\right]\vphi\left(\tilde{\nu}^{(i)}, \tilde{\tau}^{(i)}\right)\vphi^{\He}\left(\tilde{\nu}^{(i)}, \tilde{\tau}^{(i)}\right) = \mpsi_{\dict}\diag\left(\mathbb{E}\left[\veta\right]\right)\mpsi_{\dict}^{\He}
\end{split}
\label{eq:th-1-cov-proof}
\end{equation}
\end{figure*}

\balance
\bibliographystyle{IEEEtran}
  
\bibliography{IEEEabrv,mybib}

\end{document}